\documentclass[submission,copyright,creativecommons]{eptcs}
\providecommand{\event}{TERMGRAPH 2022} 
\usepackage{breakurl}             
\usepackage{underscore}           

\usepackage{amssymb}
\usepackage{graphics}
\usepackage{xcolor}
\usepackage{mathtools}
\usepackage{amsthm}
\usepackage{stmaryrd}

\definecolor{grey}{cmyk}{0.0,0.0,0.0,0.35}

\newtheorem{definition}{Definition}
\newtheorem{remark}{Remark}
\newtheorem{example}{Example}
\newtheorem{lemma}{Lemma}
\newtheorem{theorem}{Theorem}

\newcounter{Romancnt}

\newcommand{\romannr}[1]{\setcounter{Romancnt}{#1}\roman{Romancnt}}

\catcode`\@=11
\newcommand{\wash}{\mathpalette\mywash}
\newcommand{\mywash}[2]{\setbox0=\hbox{$\m@th#1{#2}$}\wd0=0pt\box0}
\catcode`\@=12

\newcommand{\arsdev}{\mathrel{\wash{\,\,\,\,\circ}\mathord{\longrightarrow}}}

\makeatletter
\def\lubd#1#2#3{%
 \def\next##1##2{%
  \setbox0=\hbox{$##1\vphantom{#3}\@ifempty{#1}{}{_{\vphantom{#1}}}%
   \@ifempty{#2}{}{^{#2}}$}%
  \setbox1=\hbox{$##1\vphantom{#3}\@ifempty{#1}{}{_{#1}}%
   \@ifempty{#2}{}{^{\vphantom{#2}}}$}%
  \setbox2=\vbox{\hbox to\wd0{}\hbox to\wd1{}}%
  {\hskip\wd2\hskip-\wd0\box0\hskip-\wd1\box1{#3}}%
 }%
 \mathpalette\next{}%
}

\def\rubd#1#2#3{#3\@ifempty{#1}{}{_{#1}}\@ifempty{#2}{}{^{#2}}}
\def\rubd#1#2#3{{\rubd{#1}{#2}{#3}}}
\makeatother

\def\ip#1#2{#1_{#2}}

\def\Rho{{P}}

\def\Chi{{X}}
\def\Tau{{T}}

\def\isdefd{\mathrel{{:}{=}}}

\def\funap#1#2{#1(#2)}
\def\binap#1#2#3{#2\mathbin{#1}#3}
\def\relap#1#2#3{#2\mathrel{#1}#3}

\def\hastype{\mathbin{:}}

\def\strlft{\langle}
\def\strrgt{\rangle}
\def\strsep{,}
\def\str#1{\strlft#1\strrgt}
\def\strtwo#1#2{\str{#1\strsep#2}}

\def\setstr#1{\{#1\}}

\def\sseteq{{=}}
\def\seteq{\relap{\sseteq}}

\def\ssetle{{\subseteq}}
\def\setle{\relap{\ssetle}}
\def\ssetge{{\supseteq}}
\def\setge{\relap{\ssetge}}

\def\Nat{\mathbb{N}}

\def\sfunrg{\funrg{\,}}
\def\sfungr{\mathsf{TS}}
\def\funrg{\mea}
\def\fungr{\funap{\sfungr}}

\def\anat{{n}}

\def\natzer{{0}}

\def\natone{{1}}
\def\nattwo{{2}}

\def\snatge{{\geq}}
\def\natge{\relap{\snatge}}

\def\aidx{{i}}

\def\astr{{s}}
\def\bstr{{t}}
\def\cstr{{u}}

\def\aistr{\ip{\astr}}
\def\bistr{\ip{\bstr}}

\def\stremp{{\varepsilon}}

\def\aArs{{\Rightarrow}}

\def\sstpsrc{\mathsf{src}}
\def\sstptgt{\mathsf{tgt}}

\def\aars{{\rightarrow}}

\def\astp{\phi}

\def\bstp{\psi}
\def\cstp{\chi}

\def\aStp{\Phi}
\def\bStp{\Psi}
\def\cStp{\Chi}

\def\arsa{\relap{\aars}}
\def\aobj{{a}}
\def\bobj{{b}}

\def\aarsdev{{\arsdev}}
\def\arsdeva{\relap{\aarsdev}}

\def\acnv{{\gamma}}

\def\aicnv{\ip{\acnv}}

\def\bcnv{{\delta}}
\def\ccnv{{\zeta}}
\def\dcnv{{\eta}}

\def\scnvsco{{\cdot}} 
\def\cnvsco{\binap{\scnvsco}}

\def\sscnveqv{{\equiv}}
\def\sscnveqva{\relap{\sscnveqv}}
\def\scnveqv{{\sscnveqv}}
\def\cnveqv{\relap{\scnveqv}}

\def\sscnvge{{\geqslant}}
\def\sscnvgea{\relap{\sscnvge}}

\def\scnvge{\sscnvge}
\def\cnvge{\relap{\scnvge}}

\def\ares{{/}}
\def\resa#1#2{#1\ares#2}

\def\stpdsc#1{\llbracket#1\rangle\kern-2.1pt\rangle}
\def\stpdscinv#1{\langle\kern-2.1pt\langle#1\rrbracket}

\def\posemp{{\varepsilon}}
\def\posone{\natone}
\def\postwo{\nattwo}

\def\vposemp{\mathring{\posemp}}
\def\eposone{\overline{\posone}}
\def\vposone{\mathring{\posone}}
\def\epostwo{\overline{\postwo}}
\def\vpostwo{\mathring{\postwo}}

\def\atrm{t}

\def\strmeq{{=}}
\def\trmeq{\relap{\strmeq}}

\def\asym{\aobj}
\def\bsym{\bobj}

\def\aisym{\ip{\asym}}

\def\aSym{{\Sigma}}

\def\bSym{{\Tau}}

\def\srulstr{{\to}}
\def\rulstr{\binap{\srulstr}}

\def\arul{\varrho}

\def\aRul{\Rho}

\def\alhs{{\ell}}
\def\arhs{{r}}

\def\mea#1{\llbracket#1\rrbracket}
\def\smea{\mea{\,}}

\def\sfi{\mathsf{i}}
\def\sfs{\mathsf{s}}
\def\sft{\mathsf{t}}

\def\sinit{\mathsf{M}}
\def\init{\funap{\sinit}}
\def\tinit{\mathsf{N}}
\def\stail{\mathsf{R}}

\definecolor{rood}{rgb}{1,0,0}
\definecolor{groen}{rgb}{0,.8,0}
\definecolor{blauw}{rgb}{0,0,1}
\definecolor{cyaan}{rgb}{0,.8,.8}
\definecolor{magenta}{rgb}{1,0,1}
\definecolor{geel}{rgb}{.9,.7,0}
\definecolor{zwart}{rgb}{0,0,0}
\newcommand{\mycolor}[2]{{\color{#1}#2}}

\def\scolone{\mycolor{blauw}}

\def\scolthr{\mycolor{rood}}
\def\A{\scolthr{A}}
\def\B{\scolone{B}}
\def\sunderline#1#2{\wash{\kern.1em\underline{\phantom{\,#1\,}}}#1#2}

\def\stpsco{\cnvsco}

\title{On Causal Equivalence by Tracing in String Rewriting}
\author{Vincent van Oostrom\thanks{Supported by EPSRC Project EP/R029121/1 Typed lambda-calculi
with sharing and unsharing. Most of this work was performed while employed
at the University of Bath, England.}
\email{oostrom@javakade.nl}
}
\def\titlerunning{On Causal Equivalence by Tracing in String Rewriting}
\def\authorrunning{V. van Oostrom}
\begin{document}
\maketitle

\begin{abstract}
We introduce proof terms for string rewrite systems and, 
using these, show that various notions of equivalence on reductions 
known from the literature can be viewed as different perspectives
on the notion of causal equivalence.
In particular, we show that \emph{permutation equivalence}
classes (as known from the $\lambda$-calculus and term rewriting) 
are uniquely represented both by \emph{trace graphs} (known from 
physics as causal graphs) and by so-called 
\emph{greedy multistep} reductions (as known from algebra).
We present effective maps from the former to the latter,
\emph{topological multi-sorting} $\sfungr$, and vice versa,
the \emph{proof term algebra} $\sfunrg$.
\end{abstract}

\section{Introduction} \label{sec:intro}

We are interested in all aspects of computations 
as modelled by rewrite systems. Here, we are interested
in \emph{finite} computations `doing the same work up 
to the order of the tasks performed'. This can be analysed from 
the perspective of \emph{causality} with the idea that it is 
exactly the causally independent tasks that can be reordered. 
In~\cite[Chapter~8]{Tere:03} we presented five conceptually
distinct ways to mathematically model \emph{causal} 
equivalence~\cite[Section~8.1.3]{Tere:03} of computations
in term rewriting, based on \emph{permutation}, \emph{labelling}, 
\emph{standardisation}, \emph{extraction} and \emph{projection}, 
respectively, and showed them to coincide.

Though coincidence of the above five perspectives 
gave us confidence in having captured 
causal equivalence within term rewriting, at the time we failed
to relate them to the further perspective put forward there based 
on \emph{tracing}~\cite[Section~8.6]{Tere:03}.
The problem\footnote{%
Cf.\ the paragraph on top of page~415 in \cite[Remark~8.6.74]{Tere:03}.}
to do so resides in term rewrite rules that are \emph{non-linear}.
For example, consider a reduction 
$f(a) \to g(a,a) \to g(b,a) \to g(b,c)$ 
in the term rewrite system having rules $a \to b$, $a \to c$ and 
$f(x) \to g(x,x)$, with the last rule non-linear (it \emph{copies} $x$).
The single occurrence of the symbol $a$ in $f(a)$ \emph{traces}\footnote{%
For the \emph{trace} relations presented in~\cite[Section~8.6.1]{Tere:03};
see Definitions~8.6.7,~8.6.17, 
Lemma~8.6.14, and Proposition~8.6.18.}
to both occurrences of $a$ in $g(a,a)$, hence \emph{causes} both subsequent steps.
However, the fact that $a$ was \emph{copied} by the first step 
is \emph{not} captured by tracing; $f$ traces to neither copy of $a$,
though the first step rewriting $f$ is the \emph{cause}
of having the two further steps rewriting the copies of $a$, in the first place.
In this paper we show that only non-linearity
is problematic for tracing. More precisely, we show that for string rewriting,
which is inherently \emph{linear},
causal equivalence does have a simple characterisation based on tracing,
namely by \emph{tragrs}.\footnote{%
Tragr is short for \emph{tra}ce \emph{gr}aph and pronounced as \emph{tracker}, 
with the idea that a tragr \emph{tracks} what happens to symbols.}

In Section~\ref{sec:srs} we adapt the theory of \emph{permutation}
equivalence of finite computations as developed 
in~\cite[Chapter~8]{Tere:03} for term rewriting, to the case 
of string rewriting at hand. 
To represent, interpret, and prove properties about, finite computations 
by algebraic and inductive means, we adapt \emph{proof} 
terms~\cite[Chapter~8]{Tere:03} to represent the finite reductions of a string rewrite system.
In Section~\ref{sec:tragr} we introduce tragrs 
as a formalisation of Wolfram's notion of \emph{causal graph}~\cite{Wolf:20}.
We relate proof terms to tragrs by on the one hand giving 
an algebra $\mea{\,}$ interpreting proof terms as tragrs 
preserving permutation equivalence. 
and on the other hand presenting a topological multi-sorting algorithm $\sfungr$ 
mapping tragrs back to proof terms. 
In Section~\ref{sec:greedy} we show that any proof term may be transformed,
by repeatedly swapping \emph{loath} pairs,\footnote{
A loath pair intuitively is a pair of consecutive steps that are causally independent, 
so that the second step could have been in
parallel to the first, but it was too \emph{loath} to do so.}
into a permutation equivalent 
\emph{greedy}~\cite{Deho:15} multistep reduction, 
and that the latter are in bijective correspondence to tragrs 
via the maps $\sfunrg$ and $\sfungr$. This then allows us to wrap up and conclude that 
tragrs (and greedy multistep reductions) serve as unique representatives
of permutation equivalence classes. 
\begin{remark} \label{rem:klop} 
This short paper was provoked by a remark made to me\footnote{%
While employed as Universit\"atsassistentIn in the
Computational Logic group at the University of Innsbruck.}
in 2020 by Jan Willem Klop, that Wolfram's causal graphs~\cite{Wolf:20}
should characterise permutation equivalence, for string rewriting. 
Having followed Wolfram's \emph{physics} project myself
and having observed that its developments frequently ran 
parallel to those in~\cite[Chapter~8]{Tere:03}, in particular that
there was a close connection between causal graphs and the trace relations 
in~\cite[Section~8.6.1]{Tere:03} discussed above, allowed me to reply immediately, 
confirming Klop's intuition, referring him to~\cite[Chapter~8]{Tere:03} 
and drawing Figure~\ref{fig:tragr}.
At the time I was reluctant to develop that further, as the idea 
was simple and did not solve the problem for the non-linear case 
left open in~\cite{Tere:03}. Later, in 2022,\footnote{%
While employed as Research Associate in the Mathematical Foundations 
of computation group at the University of Bath.} 
I realised that a single picture was too cryptic, and that 
only \emph{because} the results are simple here do we 
entertain hope to extend them to the complex non-linear case.
\end{remark}

\section{Proof terms for string rewriting} \label{sec:srs}

We adapt the theory of permutation equivalence from term 
rewriting~\cite[Chapter~8]{Tere:03} to string rewriting,
guided by that strings can be represented as terms so that 
extant theory for term rewriting can be adapted to string rewriting.
String rewriting affords better properties than term rewriting due to linearity:
whereas term rewrite steps may be \emph{non-linear} as they 
can \emph{replicate}, erase or copy, subterms 
of arbitrary sizes, string rewrite steps cannot do so; they are \emph{linear}.
We moreover forbid left- and right-hand sides of string rewrite rules to be the empty string.
This restriction makes sense from the perspective of causality~\cite{Wins:89}
as it entails all steps being \emph{ex materia} 
(forbidding \emph{ex nihilo} steps) and having \emph{bounded} 
causes; cf.\ the $4^\text{th}$ item of Remark~\ref{rem:empty}.
By imposing these (linearity and non-emptiness) restrictions, we 
are in a sweet spot; the resulting string rewrite systems have \emph{sufficient}
structure to \emph{express} the different perspectives on causal equivalence
mentioned in the introduction, and these perspectives can in turn be proven equivalent in 
a \emph{simple} way due to the absence of replication. 
As in~\cite[Chapter~8]{Tere:03}, to state and prove 
results \emph{proof} terms are our tool of choice
for representing the reductions of a string rewrite system.

The usual definition of the finite strings over an alphabet $\aSym$ as the free 
monoid over $\aSym$ is abstract. To be able to deal with matters of representation,
we instead will be concrete here.
\begin{definition}
  A term rewrite system is \emph{oudenadic} if all rule and function symbols
  have arity $\natzer$, it has a nullary symbol $\stremp$ (empty string) and
  a binary composition symbol $\cdot_h$, and terms are considered modulo $\equiv_M$ induced 
  by the monoid laws, i.e.\ $\trmeq{\stremp\cdot_h\astr}{\astr}$, 
  $\trmeq{\astr\cdot_h\stremp}{\astr}$,
  and $\trmeq{(\astr\cdot_h\bstr)\cdot_h\cstr}{\astr\cdot_h(\bstr\cdot_h\cstr)}$.
\end{definition}
\begin{remark} \label{rem:oudenadic}
  Our terminology \emph{oudenadic} is an attempt to highlight that the representation 
  employed here associates \emph{nullary} function symbols to letters,
  and to contrast it with the usual \emph{monadic} representation,
  which associates a \emph{unary} function symbol to each letter;
  cf.~\cite[Section~3.4.4]{Tere:03} for an account of both.

  In our modelling, strings are closed oudenadic terms over the alphabet
  modulo the monoid laws. Uniquely representing such equivalence classes
  can itself be achieved by term rewriting: orienting the above monoid laws from 
  left to right yields a complete (confluent and terminating) term rewrite system, 
  having as normal forms strings of shape either $\stremp$ or 
  $\aisym{\natone}\ldots\aisym{\anat}$ for some $\natge{\anat}{\natone}$.
\end{remark}
We refer to $\cdot_h$ as \emph{horizontal} composition
  to distinguish it from \emph{vertical} composition $\cdot_v$ 
  below (see Definition~\ref{def:proofterm}), 
based on that we adhere to a convention of drawing strings 
  horizontally and reductions vertically in figures. That meshes 
  well with thinking of strings as being extended in space ($1$-dimensional,
  horizontally) and of reductions as extended in time (vertically);
  cf.\ Figure~\ref{fig:evolution}.
We assume $\cdot_h$ is infix and right-associative and that it is left implicit,
i.e.\ is represented by juxtaposition. To that end, 
we assume $\aSym$ has \emph{unique reading}: if 
$a_1\ldots a_n = b_1\ldots b_m$ for $a_i,b_j \in \aSym$, then
$n = m$ and $a_k = b_k$ for all $1 \leq k \leq n$, cf.~\cite{Viss:11}.
\begin{example} \label{exa:string}
  The alphabet $\aSym \isdefd \setstr{\A,\B}$ has unique reading.
  Per our conventions $\A\B\A\A\B$ abbreviates the term 
  $\A \cdot_h (\B \cdot_h (\A \cdot_h (\A \cdot_h \B)))$,
  which is closed and in normal form with respect to the monoid rules,
  so serves as the unique representative of the string
  (an $\equiv_M$-equivalence class containing, e.g., $(\A\B)(\A\A)\B$).
  
  The alphabet $\bSym \isdefd \setstr{AB,B,A}$ does not have unique
  reading, e.g.\ $ABB$ can be viewed as being composed of  
  the two letters $AB$ and $B$, and alternatively 
  of the three letters $A$, $B$ and $B$.
\end{example}
Concretely, a \emph{string} rewrite system \emph{over} an alphabet $\aSym$ is 
an oudenadic term rewrite system having the letters in $\aSym$ as nullary function symbols,
with sources and targets of rules being nonempty strings (cf.\ the introduction),
and steps taking place modulo $\equiv_M$.
We use $a,b,\ldots$ as variables for letters, $A,B,\ldots$ as concrete letters,
and $\A,\B$ for the concrete letters of our running example, as
in Example~\ref{exa:string}.

We consider term rewrite systems in the sense of~\cite[Chapters~8 and~9]{Tere:03},
meaning that rules themselves will feature as \emph{symbols} whose arity ($\natzer$ for 
the oudenadic systems we consider here) is the number
of variables in the rule, and rules come
equipped with source / target functions mapping them to their lhs / rhs.
This enables expressing reductions, and more generally proofs in rewrite logic~\cite{Mese:92},
as \emph{proof} terms~\cite{Tere:03}, terms over a signature comprising the letters, 
the rules, and a binary composition symbol $\scnvsco_v$ representing the transitivity inference 
rule of rewrite logic~\cite{Mese:92}.
In turn, this allows us to represent the key notion of this paper,
the notion of causal equivalence, as an equivalence on reductions / proof \emph{terms}. 
\begin{definition} \label{def:proofterm}
  Consider for an oudenadic term rewrite system $\strtwo{\aSym}{\aRul}$,
  the extension of the  oudenadic signature for $\aSym$
  by the rules $\arul$ in $\aRul$ as nullary symbols
  and the binary composition symbol $\cdot_v$.  
  \emph{Proof} terms are a subset 
  of the terms over this signature defined inductively, 
  together with \emph{source} and \emph{target} functions $\sstpsrc$ 
  and $\sstptgt$ to strings, in Table~\ref{tab:str}, where
  we use $\acnv \hastype \cnvge{\astr}{\bstr}$ to denote that $\acnv$ is a proof
  term having string $\astr$ as source and string $\bstr$ as target,
  and employ $\acnv,\bcnv,\ccnv,\dcnv,\ldots$ to range over proof terms.
\end{definition}
\begin{table}[htb]
  \[\begin{array}{l@{\quad}r@{\,\hastype\,}r@{\,\scnvge\,}l@{\quad}l}
    \text{(empty)} &
    \stremp & \stremp & \stremp \\
    \text{(letter)} &
    \asym & \asym & \asym & \text{for each letter $\asym$} \\
    \text{(rule)} &
    \arul & \alhs & \arhs & \text{for each rule $\arul \hastype \rulstr{\alhs}{\arhs}$} \\
    \text{(juxtaposition)} &
    \aicnv{\natone}\aicnv{\nattwo} & \aistr{\natone}\aistr{\nattwo} & \bistr{\natone}\bistr{\nattwo} &
    \text{if $\aicnv{\aidx} \hastype \cnvge{\aistr{\aidx}}{\bistr{\aidx}}$} \\
    \text{(transitivity)} &
    \cnvsco{\acnv}{\bcnv} & \astr & \cstr & \text{if $\acnv \hastype \cnvge{\astr}{\bstr}$, 
    and $\bcnv \hastype \cnvge{\bstr}{\cstr}$} \\
  \end{array} \]
  \caption{Proof terms for string rewriting} \label{tab:str}
\end{table}
We abbreviate vertical composition $\cdot_v$ to $\cdot$, 
assume it is right-associative, and that it binds weaker than horizontal 
composition $\cdot_h$ / juxtaposition.
\begin{remark} \label{rem:al}
  By the vertical composition being on \emph{strings} 
  the target of $\acnv$ is only required to be equivalent modulo
  the monoid laws to the source of $\bcnv$ in $\text{(transitivity)}$.
  We have $\atrm \hastype \cnvge{\atrm}{\atrm}$ for every oudenadic term $\atrm$.
\end{remark}
The name proof term for such terms is justified by that they can
be viewed as a \emph{proof} that their target string is \emph{reachable} 
from their source string by using the rewrite rules.
Building on Example~\ref{exa:string}, we take the following
as a running example to illustrate concepts and results.
\begin{example} \label{exa:srs} \label{exa:proofterm}
  Let $\strtwo{\aSym}{\aRul}$ be the string rewrite system having rules
  $\aRul \isdefd \setstr{\alpha \hastype \rulstr{\B\B}{\A}, \beta \hastype \rulstr{\A\A\B}{\B\A\A\B}}$.
  The proof term
  $\acnv \isdefd  \A\B\beta \cdot
   \A\alpha \A\A\B \cdot
   \A\A\beta \cdot
   \beta \A\A\B \cdot
   \B\beta \A\A\B \cdot
   \alpha \A\A\B\A\A\B \cdot
   \A\beta \A\A\B$ proves the \emph{reachability} statement 
   $\cnvge{\A\B\A\A\B}{\A\B\A\A\B\A\A\B}$.
  An alternative witness to that 
  statement is the proof term
  $\acnv' \isdefd \A\B\beta \cdot
 \A\alpha\beta \cdot
 \beta \A\A\B \cdot
 \B\beta \A\A\B \cdot
 \alpha\beta \A\A\B$. 
 For the (vertical) compositions in these proof terms to be well-defined it
 is essential to work modulo the monoid laws. For instance,
 although the target $(\B\A\A\B)\A\A\B$ of $\beta \A\A\B$ and the source $\B(\A\A\B)\A\A\B$
 of $\B\beta \A\A\B$ are distinct as oudenadic terms,
 they are both represented by the string $\B\A\A\B\A\A\B$,
 allowing their vertical composition in $\acnv$.
\end{example}
Although in the example the proof terms $\acnv$ and $\acnv'$ intuitively 
do `the same amount of work', the latter is shorter than the former.
This is due to that the former is maximally
sequentialised, performing one step at the time, whereas the
latter is maximally concurrent, performing steps as soon as possible
as concurrency permits.%
\begin{definition} \label{def:multistep}
  A \emph{multistep} is a proof term without vertical compositions.
  It is \emph{empty} / a (\emph{single}) \emph{step} if it has no / one occurrence of a rule.
  A \emph{reduction} (a \emph{multistep reduction}) either is an empty multistep
  or a vertical composition, associated to the right, of steps (resp.\ nonempty multisteps).
  \emph{Permutation} equivalence $\scnveqv$ between proof terms 
  is induced by the laws in Table~\ref{tab:str:rl},
  where the sides of the laws are restricted to proof terms,
  i.e.\ sources and targets of the proof terms $\acnv,\bcnv,\ccnv,\dcnv$ 
  and the oudenadic terms $\astr,\bstr$ are assumed to match appropriately. 
\end{definition}%
\begin{table}
  \[\begin{array}{l@{\quad}r@{\,=\,}l@{\quad}l@{\quad}r@{\,=\,}l}
   \text{(h-left unit)} &
    \stremp\acnv & \acnv 
    & \text{(v-left unit)} & \cnvsco{\astr}{\acnv} & \acnv \\
    \text{(h-right unit)} &
    \acnv\stremp & \acnv  
    & \text{(v-right unit)} & \cnvsco{\acnv}{\bstr} & \acnv \\
    \text{(h-associativity)} &
    (\acnv\bcnv)\ccnv & \acnv(\bcnv\ccnv) 
    & \text{(v-associativity)} &  \cnvsco{(\cnvsco{\acnv}{\bcnv})}{\ccnv} & \cnvsco{\acnv}{(\cnvsco{\bcnv}{\ccnv})} \\
    \text{(exchange)} &  \cnvsco{\acnv\bcnv}{\ccnv\dcnv} & (\cnvsco{\acnv}{\ccnv})(\cnvsco{\bcnv}{\dcnv})
  \end{array} \] 
  \caption{Laws inducing permutation equivalence} \label{tab:str:rl}
\end{table}%
We use $\aStp,\bStp,\cStp,\ldots$ to range over multisteps, and $\astp,\bstp,\cstp,\ldots$ to range over steps. Observe that the source / target of the 
left- and right-hand side of each 
law in Table~\ref{tab:str:rl} are the same (as strings).
\begin{remark} \label{rem:match}
  The requirements on the sources and targets in the laws of Table~\ref{tab:str:rl}
  boil down to working in a \emph{typed} algebraic structure~\cite{Pous:10}.
  For instance, in that setting a \emph{category} is a typed \emph{monoid},
  allowing composition of morphisms only if their sources and targets match.
  In Table~\ref{tab:str:rl} the requirements are most prominent in the (exchange) law. 
  For example, for rules 
  $\alpha  \hastype \arsa{A}{A'C},
   \alpha' \hastype \arsa{A'}{A''},
   \beta   \hastype \arsa{B}{B'},
   \beta'  \hastype \arsa{CB'}{B''}$,
  though we do have $\cnvsco{\alpha\beta}{\alpha'\beta'} \hastype \cnvge{AB}{A''B''}$,
  the expression $(\cnvsco{\alpha}{\alpha'})(\cnvsco{\beta}{\beta'})$ is
  not even a proof term since, e.g., the target $A'C$ of $\alpha$ does
  not match the source $A'$ of $\alpha'$.
\end{remark}
\begin{remark}
  Our reductions, as proof terms of a specific shape, 
  are formally distinct from the classical notion of a reduction,
  as a finite sequence of steps, in rewriting~\cite{Baad:Nipk:98,Tere:03}.
  However, since there is an obvious bijection between 
  both we feel the confusion is acceptable.
  For instance, the proof term 
   $\acnv \isdefd  \A\B\beta \cdot
   \A\alpha \A\A\B \cdot
   \A\A\beta \cdot
   \beta \A\A\B \cdot
   \B\beta \A\A\B \cdot
   \alpha \A\A\B\A\A\B \cdot
   \A\beta \A\A\B \hastype 
   \cnvge{\A\B\A\A\B}{\A\B\A\A\B\A\A\B}$ corresponds to:
   \[\arsa{\A\B\underline{\A\A\B}
}{\arsa{\A\underline{\B\B}\A\A\B
}{\arsa{\A\A\underline{\A\A\B}
}{\arsa{\underline{\A\A\B}\A\A\B
}{\arsa{\B\underline{\A\A\B}\A\A\B
}{\arsa{\underline{\B\B}\A\A\B\A\A\B
}{\arsa{\A\underline{\A\A\B}\A\A\B
      }{\A\B\A\A\B\A\A\B
      }}}}}}}\]
   Similarly, the proof term 
   $\acnv' \isdefd \A\B\beta \cdot
 \A\alpha\beta \cdot
 \beta \A\A\B \cdot
 \B\beta \A\A\B \cdot
 \alpha\beta \A\A\B$ corresponds to the following sequence of multisteps, 
 where we employ the notation $\aarsdev$ of~\cite[Chapter~8]{Tere:03} for multisteps:
 \[
  \arsdeva{\A\B\underline{\A\A\B} 
}{\arsdeva{\A\sunderline{\B}{\B}\sunderline{\A\A}{\B}
}{\arsdeva{\underline{\A\A\B}\A\A\B
}{\arsdeva{\B\underline{\A\A\B}\A\A\B
}{\arsdeva{\sunderline{\B}{\B}\sunderline{\A\A}{\B}\A\A\B
      }{\A\B\A\A\B\A\A\B
      }}}}} 
 \]
\end{remark}
Logicality, cf.~\cite{Oost:04}, of reductions expresses 
that if a reachability statement holds then it is provable by a 
\emph{reduction} that is permutation equivalent to the original proof term. 
\begin{lemma}[Logicality] \label{lem:logicality}
  If $\acnv \hastype \cnvge{\astr}{\bstr}$ for some proof 
  term $\acnv$, then there is a reduction $\acnv' \hastype \cnvge{\astr}{\bstr}$
  with $\cnveqv{\acnv}{\acnv'}$.
\end{lemma}
\begin{proof}
  By induction and cases on $\acnv$.
  \begin{list}{}{}
  \item[(empty)] 
    the empty string $\stremp$ is an empty reduction;
  \item[(letter)]
    a single letter $\asym$ is an empty reduction;
  \item[(rule)] 
    a single rule $\arul$ is a single step reduction from its lhs to its rhs;
  \item[(juxtaposition)] 
    suppose we have a proof term $\acnv \isdefd \aicnv{\natone}\aicnv{\nattwo}
    \hastype \cnvge{\aistr{\natone}\aistr{\nattwo}}{\bistr{\natone}\bistr{\nattwo}}$
    with $\aicnv{\aidx} \hastype \cnvge{\aistr{\aidx}}{\bistr{\aidx}}$.
    By the IH we have reductions 
    $\aicnv{\aidx}' \hastype \cnvge{\aistr{\aidx}}{\bistr{\aidx}}$
    with $\cnveqv{\aicnv{\aidx}}{\aicnv{\aidx}'}$.
    Set $\acnv'$ to
    $\cnvsco{\aicnv{\natone}'\langle\aistr{\nattwo}\rangle
           }{\langle\bistr{\natone}\rangle\aicnv{\nattwo}'
           }$, where for a reduction $\ccnv$ and string $\cstr$,
    $\ccnv\langle\cstr\rangle$ denotes the reduction obtained 
    by suffixing each step of $\ccnv$ by $\cstr$, and 
    symmetrically for $\langle\cstr\rangle\ccnv$.
    One easily verifies 
    $\acnv' \hastype \cnvge{\aistr{\natone}\aistr{\nattwo}}{\bistr{\natone}\bistr{\nattwo}}$,
    and also that $\cnveqv{\acnv}{\acnv'}$ by using (exchange) and vertical units repeatedly.
    Then by repeated vertical associativity applied to $\acnv'$ we obtain
    a reduction, except in case one or both of the $\aicnv{\aidx}'$
    is the empty reduction in which case we conclude by eliding one such by a horizontal unit.
  \item[(transitivity)]
    by vertically composing the reductions obtained by the IH
    for the constituent proof terms, possibly followed by associating to right and
    eliding empty reductions as before. \qedhere
  \end{list}
\end{proof}
The proof is effective, transforming proof
terms into reductions witnessing the same reachability.
\begin{example} \label{exa:perm}
  The procedure underlying the proof of Lemma~\ref{lem:logicality}
  transforms the proof term (in fact a multistep reduction) $\acnv'$ of 
  Example~\ref{exa:proofterm} into the reduction $\acnv$.
  To see this it suffices, since vertical compositions transform 
  homomorphically, to note that the 
  multisteps $\A\alpha\beta$ and $\alpha\beta \A\A\B$ in $\acnv'$ are transformed into
  the (two-step) reductions $\cnvsco{\A\alpha \A\A\B}{\A\A\beta}$ and
  $\cnvsco{\alpha \A\A\B\A\A\B}{\A\beta \A\A\B}$ in $\acnv$, respectively.  
\end{example}
\begin{remark} \label{rem:monadic}
  Logicality is the raison d'\^etre for the field of rewriting~\cite{Baad:Nipk:98,Tere:03},
  allowing to reduce reachability to reducibility.
  Cf. \cite[Lemma~3.6]{Mese:92} for the corresponding logicality result
  for term rewriting. 
\end{remark}  
Although the logicality lemma allows
to represent any proof term by a reduction, the latter is in 
general far from unique (up to permutation equivalence).
For instance, in Example~\ref{exa:perm} we could have chosen to 
transform the multistep $\A\alpha\beta$ in $\acnv'$ into the two step 
reduction $\cnvsco{\A\B\B\beta}{\A\alpha \B\A\A\B}$ instead,
giving rise to a reduction permutation equivalent to $\acnv'$ but distinct from $\acnv$.
Intuitively this is caused by that factorising a proof term into a sequence 
of steps forces to order steps in \emph{some} (arbitrary) way 
even though they may be causally independent.
For instance, $\alpha$ and $\beta$ in the multistep $\A\alpha\beta$ are concurrent /
causally independent, but still must be ordered to 
obtain a reduction; both orders will do.
Such a representation favours sequentiality over concurrency and length over width, so to speak.
In the next sections we will go in the opposite 
direction, maximally favouring concurrency over sequentiality and width over length.

From that perspective, the proof term 
$\acnv \isdefd \A\B\beta \cdot
   \A\alpha \A\A\B \cdot
   \A\A\beta \cdot
   \beta \A\A\B \cdot
   \B\beta \A\A\B \cdot
   \alpha \A\A\B\A\A\B \cdot
   \A\beta \A\A\B$ is a proof of the reachability statement 
   $\cnvge{\A\B\A\A\B}{\A\B\A\A\B\A\A\B}$ that is wasteful in two ways:
\begin{list}{}{}
\item[(too long)] This can be remedied by proceeding greedily~\cite{Deho:15},
  employing proper \emph{multi}steps instead of steps.
  For instance, the $2^\text{nd}$ and $3^\text{rd}$ steps 
  $ \A\alpha \A\A\B \cdot \A\A\beta \hastype \sscnvgea{\A\B\B\A\A\B}{\A\A\B\A\A\B}$
  in $\acnv$ can be combined into the single multistep
  $\A\alpha\beta \hastype \sscnvgea{\A\B\B\A\A\B}{\A\A\B\A\A\B}$.
  Proceeding greedily, combining as many of the single
  steps into multisteps as possible, and as early as possible,
  turns $\acnv$ into the shorter \emph{greedy} multistep reduction $\acnv' \isdefd
  \A\B\beta \cdot
  \A\alpha\beta \cdot
  \beta \A\A\B \cdot
  \B\beta \A\A\B \cdot
  \alpha\beta \A\A\B$.
  As we will show, greedy multistep reductions may serve as \emph{unique}
  representatives of permutation equivalence classes.
\item[(too large)]
  (Multi)steps not only represent what changes (via the rules in it) but also 
  what \emph{does not change} (via the letters in it); cf.\ the frame problem~\cite{Shan:16}.
  As a consequence, in general proof terms predominantly consist of letters;
  this holds true in particular both for $\acnv$ and $\acnv'$.
  \emph{Causal graphs}~\cite{Wolf:02} (cf.\ Figure~\ref{fig:tragr} left) remedy this
  by eliding letters, only keeping the causal dependencies between rule symbols.
  This suffices, as we will show, to let causal graphs serve as \emph{unique}
  representatives of permutation equivalence classes.
\end{list}
To express and relate both remedies we will employ a bit of \emph{residual} theory 
(going back to~\cite{Chur:Ross:36}) for multisteps below. 
To avoid things becoming too heavy for 
this short paper, we only develop the residual theory necessary 
here and in an ad hoc informal fashion, referring the reader 
to Chapter~8 of~\cite{Tere:03} in general and to Section~8.7
in particular, for background on (from the perspective
of permutation equivalence) and a formal treatment of, residuation.
\begin{definition} \label{def:res}
  For multisteps $\aStp$, $\bStp$ having the same source,
  we write $\setle{\aStp}{\bStp}$ to denote that $\aStp$ is \emph{contained} in $\bStp$,
  meaning that $\aStp$ is obtained from $\bStp$ by mapping \emph{some} occurrences of 
  rule symbols to their source.
  In that case, we denote by $\resa{\bStp}{\aStp}$ the \emph{residual} of $\bStp$ \emph{after} $\aStp$,
  that is, the multistep obtained from $\bStp$ by mapping the \emph{other} occurrences of rules 
  (the \emph{complement} of those selected for $\setle{\aStp}{\bStp}$) to their target.
\end{definition}
\begin{example}
  $\A\B\B\A\A\B$, $\A\B\B\beta$, $\A\alpha \A\A\B$ and $\A\alpha \beta$ are
  the four multisteps contained in $\A\alpha\beta$ in Example~\ref{exa:proofterm}.
  We have, e.g., $\trmeq{\resa{\A\alpha\beta}{\A\B\B\beta}}{\A\alpha \B\A\A\B}$ and
  $\trmeq{\resa{\A\alpha\beta}{\A\alpha \A\A\B}}{\A\A\beta}$.
  Observe that if $\setle{\aStp}{\bStp}$ and $\aStp$ is nonempty, then 
  fewer rule symbols occur in $\resa{\bStp}{\aStp}$ than in $\bStp$ by linearity
  of string rewriting.
\end{example}

\section{Trace graphs by proof term algebra} \label{sec:tragr}

We give a proof term algebra $\mea{\,}$ into tragrs,
trace graphs, based on causal graphs~\cite{Wolf:02}.
The algebra is shown to model permutation equivalence in that
it maps permutation equivalent proof terms to the same tragr.
We give a procedure we dub topological multi-sorting,
reading back a proof term from a tragr.

Before giving a formal treatment, we first
give some underlying intuitions by means of an example 
that links to the intermediate informal 
notion of an \emph{evolution}~\cite{Wolf:02},
and to our discussion above.%
\begin{figure}
\def\afig{$\A\B\A\A\B$}
\def\bfig{$\A\B\B\A\A\B$}
\def\cfig{$\A\A\A\A\B$}
\def\dfig{$\A\A\B\A\A\B$}
\def\efig{$\B\A\A\B\A\A\B$}
\def\ffig{$\B\B\A\A\B\A\A\B$}
\def\gfig{$\A\A\A\B\A\A\B$}
\def\hfig{$\A\B\A\A\B\A\A\B$}
\def\abfig{$\A\B\beta$}
\def\bcfig{$\A\alpha\A\A\B$}
\def\cdfig{$\A\A\beta$}
\def\defig{$\beta\A\A\B$}
\def\effig{$\B\beta\A\A\B$}
\def\fgfig{$\alpha\A\A\B\A\A\B$}
\def\ghfig{$\A\beta\A\A\B$}
\def\ifig{$\B\B$}
\def\jfig{$\A$}
\def\kfig{$\A\A\B$}
\def\lfig{$\B\A\A\B$}
\def\ijfig{$\alpha$}
\def\klfig{$\beta$}
 \begin{center}
\scalebox{1.58}{\includegraphics{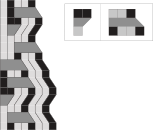}}\quad\quad\quad\quad%
\scalebox{.57}{\input{evolutionwolf.txt}}
  \end{center}
  \caption{Reduction $\acnv \hastype \cnvge{\A\B\A\A\B}{\A\B\A\A\B\A\A\B}$ (right)
  and its evolution (left)} \label{fig:evolution}
\end{figure}%
\begin{example} \label{exa:evolution}%
\def\rulea{\raisebox{-.12ex}{\scalebox{.41}{\includegraphics{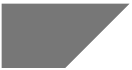}}}}%
\def\ruleb{\raisebox{-.12ex}{\scalebox{.41}{\includegraphics{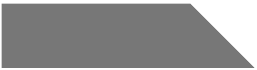}}}}%
\def\abruleb{\raisebox{-.12ex}{\scalebox{.39}{\includegraphics{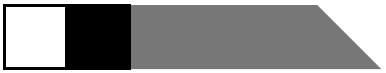}}}}%
\def\ql{\kern-.15em}%
  The reduction $\acnv$ of Example~\ref{exa:proofterm} can be depicted
  as the evolution on the left of in Fig.~\ref{fig:evolution} 
  (taken from~\cite{Rowl:Weis}; based on~\cite[fig.~a, p.498]{Wolf:02}).
  To that end, we interpret steps as rows of (possibly skewed) blocks 
  obtained by $\A \mapsto \Box$, $\B \mapsto \blacksquare$, $\alpha \mapsto \rulea$, and $\beta \mapsto \ruleb$.
  Vertical compositions of steps are interpreted by stacking the rows of the interpretations
  of the steps on top of each other, interspersed with the evaluations 
  of their sources and targets.
  For instance, the top three rows are the evaluations 
  $\Box\ql\blacksquare\ql\Box\ql\Box\ql\blacksquare$,
  $\abruleb$, and
  $\Box\ql\blacksquare\ql\blacksquare\ql\Box\ql\Box\ql\blacksquare$
  of the source, step, and target of $\A\B\beta \hastype \sscnvgea{\A\B\A\A\B}{\A\B\B\A\A\B}$.
  (The interpretations of the rule symbols $\alpha,\beta$ are given next to the evolution).
\end{example}
Looking at Figure~\ref{fig:evolution} the correspondence
between evolutions and reductions is clear though informal.
Evolutions nicely illustrate the point argued above in (too large) that
letters (the white and black boxes representing $\A$ and $\B$) add nothing
to the representation; the source and target strings and the causal dependencies 
(represented by directed edges) between the rule symbols would suffice to 
read back the multistep
reduction $\acnv'$ (permutation equivalent to $\acnv$) from the evolution. That idea 
will be formalised below using the notion of \emph{tragr}, 
short for \emph{trace graph}, illustrated for $\acnv$ / $\acnv'$ in Figure~\ref{fig:tragr}.
\begin{figure}
\def\afig{$\sfs$}
\def\bfig{$\sfi$}
\def\cfig{$\sft$}
\def\dfig{$\stremp$}
\begin{center}\scalebox{.65}{\includegraphics{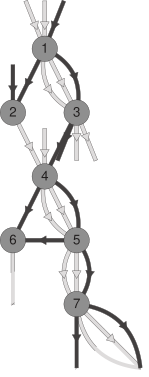}\quad\quad\quad\quad\quad\quad\quad\quad\scalebox{.47}{\input{tragr2.txt}}}%
\end{center}
  \caption{Causal graph (left) and tragr 
  from $\A\B\A\A\B$ to $\A\B\A\A\B\A\A\B$ (right)} \label{fig:tragr}
\end{figure}
\begin{remark}
  The book~\cite{Wolf:02} being intended for a general audience,
  causal graphs are not sufficiently formalised there to state our results here;
  in particular, causal graphs lack what we call below an interface 
  (dags of the source and target strings).
  Tragrs are our way to overcome that deficiency.
  We believe that if Wolfram were to formalise his notion
  of causal graph, he would end up with something similar to tragrs.
\end{remark}
\begin{definition} \label{def:tragr}
  Given a string rewrite system $(\aSym,\aRul)$, a 
  \emph{tragr from string $\astr$ to string $\bstr$}
  is a port graph~\cite{Stew:02} (see also~\cite{Bawd:86,Lafo:90,Stew:02,Oost:Looi:Zwit:04}).
  comprising the following three parts, as visualised in:
\begin{center}%
\scalebox{.33}{\def\efig{$\stremp$}\input{tracegraph.txt}}
\end{center}
  \begin{itemize}
  \item
    the dag of source string $\astr$
    having for every occurrence of a letter $\bsym$ in $\astr$ a node labelled $\bsym$,
    having (in clockwise order) an input port of type\footnote{%
Types serve here only to enable indicating / visualising
connections between ports and edges conveniently (cf.\ Example~\ref{exa:tragr}).}
    $\ast$, an output port of type $\ast$,
    and an output port of type $\bsym$.
    The nodes are connected in a straight line by edges of type $\ast$,
    terminated by a node labelled $\stremp$ having an input port of type $\ast$,
    and an output port of type $\stremp$.
  \item
    a dag, the \emph{causal} graph, of nodes 
    labelled by rule symbols $\arul$ having (in clockwise order) 
    as input ports the letters of the source string of $\arul$ and as output
    ports the letters of (the reverse of) the target string, with
    each port having the type of its letter;
  \item
    the dag of target string $\bstr$,
    as for the source string but in reverse direction, i.e.\ 
    with input and output port of type $\ast$ swapped.
  \end{itemize}
  The tragr is required to be a planar dag,
  to only have edges from input to
  output ports of the same type, and to have exactly two ports without edges, 
  both of type $\ast$: the first input port of the source string and the first output 
  port of the target string. (See Remark~\ref{rem:planar} for more on the planarity requirement.)  
\end{definition}
We indicate \emph{the} input / output ports of a tragr by dangling edges, and 
refer to the dags of the source and target strings combined as its \emph{interface}.
\begin{example} \label{exa:tragr}
  The graph on the right in Figure~\ref{fig:tragr} is a tragr
  with the types of edges being indicated by color.
\end{example}
\begin{remark} 
  Tragrs are not (too large) in the sense discussed above; letters
    only feature in the interface but not in the causal graph of a 
    tragr; cf.\ the text below~\cite[Def.~8.6.17]{Tere:03}.
\end{remark}
\begin{definition} \label{def:tragr:alg}
  For a string rewrite system $(\aSym,\aRul)$ the 
  proof term algebra $\mea{\,}$ on tragrs,
  interpreting each proof term $\acnv \hastype \cnvge{\astr}{\bstr}$ 
  as a tragr $\mea{\acnv}$ from $\astr$ to $\bstr$,
  is given by:
\def\afig{$\asym$}%
\def\bfig{$\Rightarrow$}%
\def\efig{$\stremp$}%
\def\sfig{$\sstpsrc$}%
\def\tfig{$\sstptgt$}%
\def\pfig{$1,\text{i}$}%
\def\qfig{$2,\text{i}$}%
\def\rfig{$1,\text{o}$}%
\begin{list}{}{}
  \item[(letter and empty)]
    $\mea{a}$ and $\mea{\stremp}$ are the tragrs:
    \begin{center}
    \scalebox{.36}{\input{alphabetagraphsymbol.txt}}
    \end{center}
  \item[(rule)]
    $\mea{\arul}$ is 
    a tragr having the straight line dags for its source and target as interface,
    comprising a single rule node connected to the interface in an orderly way,
    illustrated for rules $\alpha$ and $\beta$ by:
    \begin{center}
      \scalebox{.36}{\input{alphabetagraphrule.txt}}
    \end{center}
  \item[(juxtaposition)]
    $\mea{\acnv\bcnv}$ is obtained from $\mea{\acnv}$ and $\mea{\bcnv}$ 
    by removing the $\stremp$s from the former, and redirecting
    the input and output of the latter accordingly:
    \begin{center}
\scalebox{.36}{\input{horizontal.txt}\quad\quad\quad\input{horizontal3.txt}}
    \end{center}
  \item[(transitivity)]
    The tragr $\mea{\cnvsco{\acnv}{\bcnv}}$ is obtained from $\mea{\acnv}$ and $\mea{\bcnv}$
    by connecting the output of the former to the input 
    of the latter, and subsequently eliding the intermediate interface:
    \begin{center}
\scalebox{.36}{\input{vertical.txt}\quad\quad\quad\quad\input{vertical1.txt}\quad\quad\quad\quad\input{vertical3.txt}}
    \end{center}
    where elision from the middle to the right is achieved by normalising with respect to the rules:
    \begin{center}
\scalebox{.36}{\input{vertical2.txt}}
    \end{center}
  \end{list}
\end{definition}
Observe that if $\acnv \hastype \cnvge{\astr}{\bstr}$ 
then $\mea{\acnv}$ indeed is a tragr from $\astr$ to $\bstr$.
\begin{example} \label{exa:blagr}
  The tragrs $\mea{\acnv}$ and $\mea{\acnv'}$ 
  of the permutation equivalent $\acnv,\acnv'$ 
  are as on the right in Figure~\ref{fig:tragr}.
\end{example}
\begin{remark} \label{rem:empty}
  \begin{itemize}
  \item
  Elision $\aArs$ is complete: terminating because the number of nodes decreases in each step and
  confluent because elision can be viewed as an interaction net rule~\cite{Lafo:90}.
  \item
$\mea{\acnv}$ is finite so that
all maximal paths in it lead from its input to its output,
using that nodes have at least one input / output port,
by the assumption that left- and right-hand sides are non-empty.
  \item
  We modelled trace graphs, tragrs, after the \emph{trace relations} 
  of~\cite[Definition~8.6.17 / Figure~8.37]{Tere:03} with the main difference between both being that 
  the latter do not allow \emph{parallel} edges between the same two nodes.
  That makes the latter unsuitable for
  our purposes here; only knowing \emph{that} a rule
  causally depends on another not \emph{how}, is in general
  not sufficient to read back proof terms.
  For instance, for rules $\rulstr{A}{BBB}$, $\rulstr{BB}{B}$ and $\rulstr{BB}{C}$,
  the reductions $\arsa{A}{\arsa{\underline{BB}B}{\arsa{BB}{C}}}$ and
  $\arsa{A}{\arsa{B\underline{BB}}{\arsa{BB}{C}}}$ induce the same trace relation,
  despite not being permutation equivalent.\footnote{%
It is interesting to compute their respective tragrs and see that / how they differ.%
}
  \item 
  If we were to allow right-hand sides to be empty as in $\alpha \hastype \rulstr{A}{\stremp}$, 
  then the number of occurrences of $\alpha$ that may cause the left-hand side of 
  another rule may be unbounded as illustrated by the multisteps of shape 
  $B\alpha^nB \hastype \cnvge{B\underline{A^n}B}{BB}$   
  for rule $\beta \hastype \rulstr{BB}{\ldots}$ and any $n \in \Nat$,
  despite that none of the $A$s trace to the rule $\beta$ (its left-hand side $BB$).\footnote{%
Cf.\ \cite[Definition~8.6.64]{Tere:03} for a hack to overcome (by reifying `emptyiness') 
the problem caused by such \emph{collapsing} rules.}
  Swapping left- and right-hand sides in this example illustrates the
  problem with allowing left-hand sides to be empty.  
  \item
  The algebra $\mea{\,}$ illustrates that 
  horizontal and vertical composition are closely
  related to parallel and series composition of graphs.
  \end{itemize}
\end{remark}
We show that $\mea{\,}$ maps permutation 
equivalent proof terms to the same tragr,\footnote{%
Stated differently, we show the proof term algebra $\smea$
is a \emph{model} (in an appropriate typed sense of~\cite[Definition~2.5.1(\romannr{5})]{Tere:03}) 
of permutation equivalence given in Table~\ref{tab:str:rl}.
The proof for trace graphs follows that 
for trace relations~\cite[Lemma~8.6.14]{Tere:03}.}
see Example~\ref{exa:blagr},
but defer showing the converse to the next section
(see Theorems~\ref{thm:bi} and~\ref{thm:cbe}).
\begin{lemma} \label{lem:mea:equiv} \label{lem:causal}
  $\smea$ maps permutation equivalent  
  proof terms to the same\footnote{%
Formally, the same up to graph isomorphism.
} tragr.
\end{lemma}
\begin{proof}
  We show for each law in Table~\ref{tab:str:rl}
  its left- and right-hand sides are mapped to the same tragr by $\mea{\,}$:
  \begin{itemize}
  \item
    For the monoid laws (h-left unit), (h-right unit) and (h-associativity)
    for horizontal composition, the former two follow from that the parallel
    composition of a tragr with $\mea{\stremp}$ on either side,
    amounts to first introducing and then immediately removing $\stremp$s.
    Associativity holds since removing the $\stremp$s and redirecting the
    respective input and output edges are local and independent actions.
  \item
    For the monoid laws (v-left unit), (v-right unit) and (v-associativity)
    for vertical composition, the former two follow from
    that for any string $\astr$, $\mea{\astr}$ is a \emph{ladder}, a tragr
    only comprising the straight line graphs of its source and target
    string, each the reverse of the other. For the sequential composition with
    a ladder on either side, elision amounts to the immediate removal
    of (the reverse of) the ladder.
    Associativity holds since elision is complete (confluent and terminating)
    and can be postponed until after
    connecting the respective input and output ports, which
    are local and independent actions.
  \item
    The (exchange) law holds by combining 
    the reasoning in the previous two items;
    combining removal of $\stremp$s with elision $\aArs$ is complete and
    can be postponed until after redirecting the input and output edges,
    which are local and independent actions; \footnote{%
For the reasoning to apply it is essential that
\emph{both} sides of (exchange) law are proof terms,
i.e.\ that sources and targets of the constituting proof
terms should match appropriately so that both sides
indeed yield tragrs; cf.~Remark~\ref{rem:match}.}
    see Figure~\ref{fig:exchange}.
     \qedhere
  \end{itemize}
\end{proof}%
\begin{figure}%
\def\efig{$\stremp$}%
\begin{center}%
\scalebox{.36}{\input{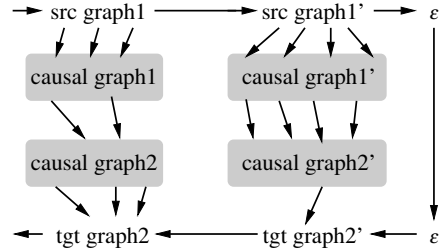}}%
\end{center}%
\caption{Tragr illustrating (exchange)} \label{fig:exchange}%
\end{figure}%
We conclude this section with showing that any tragr 
can be read back into a multistep reduction, by means
of a procedure that we dub topological \emph{multi}sorting,
which is analogous to topological sorting but selects in each
stage \emph{all} minimal elements, instead 
of just a \emph{single} one, cf.\ \cite{Plot:Prat:96}.
\begin{definition} \label{def:tm}
  The \emph{topological multi-sorting} function $\sfungr$ 
  mapping a tragr from $\astr$ to $\bstr$ 
  to a multistep reduction having $\astr$ as source and $\bstr$ as target,
  is defined by induction on size  and cases on its causal graph.
  
  If the causal graph is empty, planarity of tragrs dictates 
  the tragr is a ladder 
  (as in the Proof of Lemma~\ref{lem:causal}; 
   cf.\ the bottom-right of Appendix~\ref{app:stages}), 
  so we have $\astr = \bstr$ and may return the empty multistep $\astr$.
  
  If the causal graph is non-empty, let its \emph{minimal layer} $\sinit$ comprise 
  its minimal nodes w.r.t.\ the partial order induced by (taking the reflexive--transitive
  closure of) the dag. To construct the 
  multistep $\aStp$ we juxtapose, starting from the input of the tragr,
  the labels (letters) of nodes in the dag for $\astr$ not covered by nodes in $\sinit$,
  interspersed with the labels (rule symbols) of those covering nodes in $\sinit$.
  Let $\astr'$ be the target of $\aStp$, and consider the tragr obtained
  by replacing for every node labelled by some rule $\arul$ in 
  $T$ the source dag of $\arul$ by its target dag.
  By planarity it follows (cf.\ the top row of Appendix~\ref{app:stages})
  that the resulting tragr is from $\astr'$ to $\bstr$.
  Therefore, it suffices to vertically compose $\aStp$ with
  the $\sfungr$-image of this tragr, which exists by the IH.
  
  In both cases, we obtain a vertical composition of multisteps 
  having $\astr$ as source and $\bstr$ as target, 
  giving rise to a multistep reduction after removing 
  a trailing empty multistep.
\end{definition}
\begin{example}
  Applying topological multi-sorting $\sfungr$ to the 
  tragr on the right in Figure~\ref{fig:tragr} 
  gives rise to the $6$ successive stages displayed in Appendix~\ref{app:stages}.
\end{example}
\begin{remark} \label{rem:planar}
  Note how the planarity requirement on tragrs of Definition~\ref{def:tragr} 
  was used in Definition~\ref{def:tm} to 
  guarantee that all tragrs having an empty causal graph 
  read back as the empty multistep on their source (and target). 
  In particular, planarity disallows `rewirings' having crossing edges.
\end{remark}
\begin{lemma} \label{lem:deconstruct}
  $\sfunrg$ after $\sfungr$ is the identity on tragrs.
\end{lemma}
\begin{proof}
  By induction on size and cases on the causal graph of a tragr from $\astr$ to $\bstr$.
  
  If the causal graph is empty then by the observation in 
  Definition~\ref{def:tm} and Remark~\ref{rem:planar}, 
  the tragr is a ladder, $\astr = \bstr$,
  and it is mapped to the empty multistep $\astr$.
  Then $\funrg{\astr}$ is \emph{the} ladder from $\astr$ to $\astr$ again, 
  as follows by induction on the length of $\astr$ 
  and Definition~\ref{def:tragr:alg},
  using (letter) / (empty) in the base case
  and (juxtaposition) in the induction step.
  
  If the causal graph is not empty, then per the construction 
  in Definition~\ref{def:tm}, it is obtained 
  by (transitivity) from the tragr from $\astr$ to $\astr'$ 
  of its minimal layer $\sinit$,
  and the tragr from $\astr'$ to $\bstr$ 
  of its remaining nodes / causal graph $\stail$.
  We conclude by that the former is obtained from
  $\mea{\aStp}$ for $\aStp \hastype \cnvge{\astr}{\astr'}$
  the multistep constructed from $\sinit$ in Definition~\ref{def:tragr:alg},
  and by the induction hypothesis for the latter.
  To see the former, one proceeds as for the empty causal graph
  additionally using (rule) in the base case.
\end{proof}
\begin{remark}
  In fact, \emph{any} way to transform a tragr 
  into a proof term by repeatedly \emph{decomposing} the tragr
  by means of the \emph{inverses} of (juxtaposition) and (transitivity),
  mirrored by \emph{composing} the corresponding proof terms by means of
  horizontal and vertical composition respectively,
  and transforming the base cases (letter), (empty) and (rule)
  in the natural way, will preserve the result (Lemma~\ref{lem:deconstruct}).
  The particular such transformation $\sfungr$ was chosen here
  because it does not just yield \emph{any} proof term 
  but a \emph{greedy multistep reduction},
  which will be essential for the unique representation purposes of the next section.
\end{remark}

\section{Greedy multistep reductions} \label{sec:greedy}

We first give a standard algorithm for transforming a proof term
into a permutation equivalent \emph{greedy} one~\cite{Deho:15},
and next show there is a bijection between such greedy 
multistep reductions and tragrs. From this we conclude,
in a semantic way, that both constitute unique representatives
of permutation equivalence classes of proof terms.

We give a novel description of greediness and 
the greedy algorithm of~\cite{Deho:15}, based on the analogy with sorting and 
standardisation~\cite{Klop:80,Tere:03} in the literature.
In sorting, (adjacent) \emph{inversions} are consecutive 
elements that are out-of-order, 
and in term rewriting, \emph{anti-standard} pairs~\cite{Klop:80,Mell:05} are 
consecutive steps in a reduction such that the latter is outside (to the left of) the former.
Such pairs of out-of-order elements are of interest since they provide a \emph{local}
characterisation both of \emph{being sorted}, i.e.\ the \emph{absence} of such pairs,
and of bringing the list / reduction \emph{closer to being sorted}, 
by \emph{permuting} the out-of-order pair. This makes both processes 
amenable to a rewriting approach, with bubblesort being an example 
for sorting and the extraction of the leftmost-contracted-redex being an 
example for standardisation~\cite{Klop:80,Khas:Glau:96,Mell:96,Tere:03,Mell:05,Brug:08}.
To make the greedy algorithm fit the mould, we define \emph{loath} pairs
as consecutive multisteps where some rule symbol in the $2^{\text{nd}}$ 
is \emph{not caused} by the $1^{\text{st}}$, so may be permuted up front, 
signalling non-greediness. 
This is phrased in terms of residuation; see Definition~\ref{def:res}.
\begin{definition} \label{def:loath} \label{def:greedy}
  A proof term is \emph{greedy} if it is a multistep reduction without loath pairs,
  where a pair $\cnvsco{\aStp}{\bStp}$ of consecutive multisteps is \emph{loath} if
  there is a step $\cStp$ co-initial with $\aStp$
  such that $\setle{\aStp}{\cStp}$ and having residual
  step $\bstp \isdefd \resa{\cStp}{\aStp}$ 
  with $\setle{\bstp}{\bStp}$.
  \emph{Swapping $\cStp$ for $\cnvsco{\aStp}{\bStp}$} 
  then results in $\cnvsco{\cStp}{(\resa{\bStp}{\bstp})}$.
  Exhaustive swapping followed by removing trailing empty multisteps 
  yields a \emph{greedy decomposition}.
\end{definition} 
\begin{example}
  The multistep reduction $\acnv'$ is greedy, but $\acnv$ isn't
  as is clear from 
  $\stpsco{\A\B\beta
}{\stpsco{\overline{\stpsco{\A\alpha\underline{\A\A\B}}{\A\A\underline{\beta}}}
}{\stpsco{\beta\A\A\B
}{\stpsco{\B\beta\A\A\B
}{\overline{\stpsco{\alpha\underline{\A\A\B}\A\A\B
      }{\A\underline{\beta}\A\A\B
      }}}}}}$, where we have
  overlined its loath pairs, and 
  underlined the rule symbols and their left-hand sides involved in swapping.
  The loath pair
  $\stpsco{\A\alpha\underline{\A\A\B}}{\A\A\underline{\beta}}$ swaps into
    $\stpsco{\A\alpha\underline{\beta}}{\A\A\underline{\B\A\A\B}}$, and 
  $\stpsco{\alpha\underline{\A\A\B}\A\A\B}{\A\underline{\beta}\A\A\B}$ swaps into
  $\stpsco{\alpha\underline{\beta}\A\A\B}{\A\underline{\B\A\A\B}\A\A\B}$.
  As one may verify, exhaustive swapping yields 
  $\cnvsco{\cnvsco{\acnv'}{\A\B\A\A\B\A\A\B}}{\A\B\A\A\B\A\A\B}$,
  hence a greedy decomposition of $\acnv$ is $\acnv'$.
  Intuitively, this is as desired since $\acnv'$ exhibits 
  maximal concurrency while performing the same tasks performed in $\acnv$.
\end{example}
By standard residual theory~\cite[Chapter~8]{Tere:03}, 
swapping yields a pair of consecutive multisteps permutation 
equivalent to the original pair, as in the example.
Moreover, the size (number of rule symbols)
of the $2^{\text{nd}}$ multistep decreases per construction, so 
swapping decreases the \emph{Sekar--Ramakrishnan 
measure}~\cite[Definition~8.5.17]{Tere:03}, 
measuring a multistep reduction by the lexicographic product of the sizes of 
the multisteps in it from tail to head. Since if necessary we may first transform a 
proof term into a permutation equivalent (single step hence multistep) \emph{reduction} 
by the Logicality Lemma~\ref{lem:logicality}, we have:
\begin{lemma} \label{lem:greedy}
  A proof term can be transformed into a
  permutation equivalent greedy multistep reduction.
\end{lemma}
\begin{remark}
  To give an idea how residual theory~\cite[Table~8.5 in Section~8.7.3]{Tere:03} 
  may be employed to show swapping preserves permutation equivalence,
  first note that
  $\setle{\aStp}{\cStp}$ entails $\resa{\aStp}{\cStp}$ is an empty multistep.
  Therefore, by commutativity of join
  $\cnveqv{\cStp
 }{\cnveqv{\cnvsco{\cStp}{(\resa{\aStp}{\cStp})}
         }{\cnvsco{\aStp}{(\resa{\cStp}{\aStp}})
         }}$.
  Similarly, $\setle{\seteq{\resa{\cStp}{\aStp}}{\bstp}}{\bStp}$ entails
  $\cnveqv{\bStp}{\cnvsco{(\resa{\cStp}{\aStp})}{(\resa{\bStp}{(\resa{\cStp}{\aStp})})}}$.
  By combining both
  $\cnveqv{\cnvsco{\aStp}{\bStp}
 }{\cnveqv{\cnvsco{\aStp}{\cnvsco{(\resa{\cStp}{\aStp})}{(\resa{\bStp}{(\resa{\cStp}{\aStp})})}}
         }{\seteq{\cnvsco{\cStp}{(\resa{\bStp}{(\resa{\cStp}{\aStp})})}
                }{\cnvsco{\cStp}{(\resa{\bStp}{\bstp})}}
         }}$.
\end{remark}
\begin{remark} \label{rem:interval}
  An efficient procedure for searching for loath pairs can be 
  based on the observation that due to linearity of string rewrite systems,
  an occurrence of either a source or target of a rule can
  be identified with a \emph{pattern} in the sense of~\cite[Definition~8.6.21]{Tere:03},  
  i.e.\ with a \emph{convex} set of positions in the tree of the string
  having vertices as boundary. Following the main idea of~\cite{Oost:20},
  to see whether $\cnvsco{\aStp}{\bStp}$ is loath, it therefore suffices to check
  whether each pattern of a \emph{source} of a rule occurring in $\bStp$
  has overlap with some \emph{target} of a rule occurring in $\aStp$.
  Since a pattern in a string simply is an \emph{interval}, characterised by the
  two vertices constituting its boundary, 
  a single top--down pass through both string-trees checking disjointness of 
  intervals via their boundaries, suffices.
  If for some pattern there is no overlap, we obtain a loath pair
  by setting $\cStp$ to $\aStp$ in which the pattern was replaced by the rule.
\begin{figure}[htb]
\begin{center}%
\def\afig{$\vposemp$}
\def\bfig{$\eposone$}
\def\cfig{$\vposone$}
\scalebox{.305}{\input{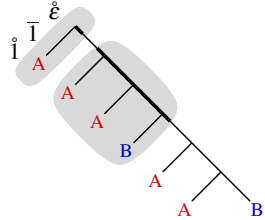}}%
\end{center}%
\caption{Non-overlapping patterns (in grey) and intervals (thick lines) of
$\A$ and $\A\A\B$ in $\underline{\A}\sunderline{\A\A}\B\A\A\B$}
\label{fig:pattern}%
\end{figure}
 For example, see Figure~\ref{fig:pattern}, using underlining to indicate patterns,
 that $\alpha \A\A\B\A\A\B,\A\beta \A\A\B$ in $\acnv$ is a loath pair follows
  from that the pattern\footnote{%
Patterns comprise node (over-ring) and edge (over-bar) positions; cf.~\cite[Example~8.6.4]{Tere:03} for more on their notation.}
  $\setstr{\vpostwo,\postwo\cdot\eposone,\postwo\cdot\vposone,\postwo\cdot\epostwo,\postwo\cdot\vpostwo,
   \postwo\cdot\postwo\cdot\eposone,
   \postwo\cdot\postwo\cdot\vposone,\postwo\cdot\postwo\cdot\epostwo,\postwo\cdot\postwo\cdot\vpostwo,
   \postwo\cdot\postwo\cdot\postwo\cdot\eposone,\postwo\cdot\postwo\cdot\postwo\cdot\vposone}$
  in $\A\underline{\A\A\B}\A\A\B$ corresponding to the source $\A\A\B$ of the rule $\beta$,
  does not have overlap with the pattern 
  $\setstr{\vposemp,\eposone,\vposone}$ 
  in $\underline{\A}\A\A\B\A\A\B$
  corresponding to the target $\A$ of the rule $\alpha$.
  This in turn follows from that the corresponding intervals
  $[\vpostwo,\postwo\cdot\postwo\cdot\vpostwo]$ and $[\vposemp]$ are disjoint
  since $\vposemp$ is smaller than $\vpostwo$.
  By disjointness / non-overlap, replacing in $\alpha \underline{\A\A\B}\A\A\B$ 
  the $\beta$-pattern $\A\A\B$
  by the rule $\beta$ yields the multistep $\alpha\beta \A\A\B$, as desired.
\end{remark}
\begin{theorem} \label{thm:bi}
  The maps $\sfunrg$ 
  and $\sfungr$ constitute a(n effective) 
  bijection 
  between greedy multistep reductions and tragrs.
\end{theorem} 
\begin{proof}
  We first show that the topological multi-sorting function $\sfunrg$ 
  maps tragrs not just to multistep reductions but to greedy such.
  Next we show that when restricting the domain of 
  the interpretation $\sfunrg$ to greedy multistep reductions,
  $\sfunrg$ and $\sfungr$ are inverse to each other
  as functions from greedy multistep reductions to tragrs and conversely.
  From that we conclude as both the functions $\sfunrg$ and $\sfungr$ are effective.

  We show that when computing the $\sfungr$-image of  
  a tragr, consecutive stages yield multisteps that are not loath 
  pairs, by induction on the number of stages.
  There is only something to show when there is more than one stage. 
  So suppose $\sfungr$ yields a composition $\cnvsco{\aStp}{\acnv}$ 
  with $\aStp$ obtained from the minimal layer of rule nodes 
  $\sinit$ of the tragr, and 
  $\acnv$ from its remaining nodes / causal graph $\stail$, 
  non-empty by assumption.
  By the IH $\acnv$ is greedy, and non-empty so has
  some first multistep, say $\bStp$, constructed from the minimal
  layer, say $\tinit$, of $\stail$.
  Per definition of $\sfungr$ each of the nodes in $\tinit$
  is reachable from some node in $\sinit$.
  Since there are no edges between the nodes in a single layers,
  this entails that for each of the nodes in $\tinit$ there
  is an edge to it from some node in $\sinit$.
  As a consequence, cf.\ Remark~\ref{rem:interval},
  the corresponding pair $\cnvsco{\aStp}{\bStp}$ 
  of consecutive multisteps is greedy / not loath.
  Thus, $\sfunrg$ maps to greedy multistep reductions.
  \begin{itemize}
  \item  
  That $\sfungr$ after $\sfunrg$ is
  the identity on greedy multistep reductions,
  we show by induction on the length of such a reduction.  
  We employ the no(ta)tions of Definition~\ref{def:tm},
  in particular we employ $\sinit$ to denote the layer of 
  minimal elements of (the causal graph of) a tragr.
  
  For the empty and single-multistep reductions this is trivial.
  Otherwise, the reduction has shape $\cnvsco{\aStp}{\acnv}$.
  By definition $\funrg{\cnvsco{\aStp}{\acnv}}$ is the serial
  composition of 
  $\funrg{\aStp}$ and $\funrg{\acnv}$ and we claim that by greediness 
  the steps in the minimal layer of the tragr $\funrg{\cnvsco{\aStp}{\acnv}}$ 
  are those of $\funrg{\aStp}$, i.e.\ 
  $\init{\funrg{\cnvsco{\aStp}{\acnv}}} = \init{\funrg{\aStp}}$.
  Then, $\aStp$ is the result of the first stage of $\sfungr$ and
  $\trmeq{\fungr{\funrg{\cnvsco{\aStp}{\acnv}}}
 }{\trmeq{\cnvsco{\aStp}{\fungr{\funrg{\acnv}}}
        }{\cnvsco{\aStp}{\acnv}
        }}$ by the IH for $\acnv$.

  It remains to prove the claim that $\init{\funrg{\cnvsco{\aStp}{\acnv}}} = \init{\funrg{\aStp}}$ for a greedy multistep reduction of shape 
  $\cnvsco{\aStp}{\acnv}$, so with $\acnv$ non-empty.
  Since $\setge{\init{\funrg{\cnvsco{\aStp}{\acnv}}}}{\init{\funrg{\aStp}}}$
  trivially holds, for arbitrary multistep reductions, suppose for 
  a proof by contradiction that 
  $\setle{\init{\funrg{\cnvsco{\aStp}{\acnv}}}}{\init{\funrg{\aStp}}}$ 
  does not hold, for $\cnvsco{\aStp}{\acnv}$ of minimal length. Then 
  there must be some node in $\init{\funrg{\acnv}}$ 
  in $\init{\funrg{\cnvsco{\aStp}{\acnv}}}$,
  per construction of $\funrg{\cnvsco{\aStp}{\acnv}}$
  as the serial composition of 
  $\funrg{\aStp}$ and $\funrg{\acnv}$.
  By minimality this node must in fact be in $\init{\funrg{\bStp}}$
  for $\bStp$ the first multistep of $\acnv$, with the node 
  corresponding to, say, step $\setle{\bstp}{\bStp}$.
  But then $\cnvsco{\aStp}{\bStp}$ would be a loath pair,
  as it allows swapping the join of $\aStp$ with $\bstp$.\footnote{%
More precisely, the join of $\aStp$ with the origin of $\bstp$ 
along the converse of $\aStp$, 
which is a step acting on an interval in the dag of 
the source string of $\aStp$, as observed in
Remark~\ref{rem:interval}.
Note our reasoning would fail if rules were allowed to have empty 
left- or right-hand sides: If $\bstp$ were due to a rule with 
an empty left-hand side, or if $\aStp$ were to contain a rule 
with an empty right-hand side, then $\Chi$ might not be swappable.
}
  This contradicts the assumed greediness of $\cnvsco{\aStp}{\acnv}$.
  \item  
  The converse direction, that $\sfunrg$ after $\sfungr$ is
  the identity on tragrs, follows from Lemma~\ref{lem:deconstruct}. \qedhere
  \end{itemize}
\end{proof}
We can now establish our main result, that
one may compute a greedy multistep reduction, unique modulo permutation equivalence,
for any proof term by first evaluating into its tragr / causal graph (using $\sfunrg$), 
followed by the topological multi-sort (using $\sfungr$) yielding the greedy
multistep reduction. 
%
\begin{theorem} \label{thm:cbe}
  For every proof term $\acnv$, there exists 
  a unique greedy multistep reduction $\acnv'$
  such that $\sscnveqva{\acnv}{\acnv'}$.
\end{theorem}
\begin{proof}
  Lemma~\ref{lem:greedy} shows existence.
  To show uniqueness, consider greedy multistep reductions $\acnv'$ and $\acnv''$
  both permutation equivalent to $\acnv$. 
  By Lemma~\ref{lem:causal}, $\mea{\acnv'}$ and $\mea{\acnv''}$ 
  are the same tragr.
  Therefore, $\acnv' = \fungr{\mea{\acnv'}} = \fungr{\mea{\acnv''}} =
  \acnv''$ by $\sfungr$ being inverse to $\sfunrg$ on greedy multistep reductions
  by  Theorem~\ref{thm:bi}..
%
%
%
%
%
%
%
%
%
\end{proof}
\begin{remark}
  \begin{itemize}
  \item
  The proof only employs one direction (the first item in the proof) 
  of Theorem~\ref{thm:bi}.
  \item
  As a consequence, using that the greedy multistep reductions \emph{are}
  the normal forms w.r.t.\ swapping, 
  we have that swapping is confluent on multistep reductions.
  This could alternatively be established via Newman's Lemma, 
  using that swapping is terminating and showing local confluence.
  \end{itemize}
\end{remark}
%
\begin{example}
  The greedy multistep reduction $\acnv'$ is the unique
  representative of the permutation equivalence class of $\acnv$.
  Both are mapped to the tragr on the right in Figure~\ref{fig:tragr} 
  by the proof term algebra $\mea{\,}$, and topological
  multi-sorting.
\end{example}

\section{Conclusions} \label{sec:conclusions}

We have shown that L\'evy's notion of permutation 
equivalence~\cite{Levy:78} as known from \emph{term rewriting}~\cite{Tere:03} 
corresponds, after specialising it to \emph{string rewriting}, 
to the notion of causal equivalence as employed by Wolfram 
in his \emph{physics} project~\cite{Wolf:02,Wolf:20}.
This we achieved by introducing trace graphs, tragrs refining Wolfram's notion of causal graph,
as representatives of permutation equivalence classes of reductions.
Representing reductions as terms themselves, so-called proof terms~\cite{Mese:92},
allowed us to specify the representation map, from reductions to tragrs,
effectively by means of a (proof term) algebra that models permutation equivalence.
To show that representatives are unique, we gave a map back 
from tragrs to so-called greedy multistep reductions as known 
from Dehornoy's work in \emph{algebra}~\cite{Deho:15},
using a topological multi-sorting procedure,
showing both maps to be inverse to each other.

The study of \emph{causality} spans all the sciences, cf.~\cite{Pear:09},
hence it is no surprise that it has been discussed and 
mathematically modelled in many ways; to mention a
few~\cite{Newm:42,Levy:78,Boud:85,Star:89,Wins:89,Joya:Stre:91,
Huet:Levy:91,Mese:92,Lane:94,Mell:96,Khas:Glau:02,Gugl:Gund:Pari:10,Hirs:13,Deho:15,Wolf:20}.
From that point of view our results can be seen as linking 
models of causality known from \emph{rewriting}~\cite{Levy:78}
(permutation equivalence), \emph{algebra}~\cite{Deho:15} (greedy multistep reductions) and
\emph{physics}~\cite{Wolf:20} (causal graphs), respectively.
In general, we think that linking different perspectives on the
\emph{same} notion, as we did for causal equivalence here
but also before in~\cite[Chapter~8]{Tere:03}, is important.
Therefore we find it surprising that in the literature mentioned
cross-references beyond the borders of the specific field 
(rewriting, algebra, physics, category theory, 
proof theory, concurrency theory,\ldots) of a paper, are few and far between. 
We hope that our short paper can contribute to creating 
at least some awareness of that, in our opinion unfortunate,
situation for causal equivalence, and the interest in overcoming it.
\begin{remark}
  Not being a physicist, I am not in a position to assess the 
  potential relation of the various causal models to physics, cf.~\cite[Section~8]{Wolf:20}.
  However, I do think it already \emph{methodologically} interesting to see how 
  far one can push causal models, which phenomena 
  can or cannot be reconstructed from them.
  For example, in~\cite{Gijs:19} it is argued that \emph{time}
  cannot be reconstructed from purely \emph{causal} models (such as rewriting),
  as the latter fail to explain \emph{synchronicity}.\footnote{%
Roughly: Why are two identical but independent clocks 
seen to be in the same state?
} That makes one wonder whether adding natural structure 
  (think of a notion of strategy or a metric) could overcome that,
  could make time emergent.
\end{remark}
The concepts and techniques developed and employed
here are simple and natural.\footnote{%
This could be the reason for the 
observed disjointedness of the literature on
causal equivalence: natural notions are
likely to be developed autonomously multiple times.
\emph{Therefore} we think it worthwhile to (try to) link such notions.
}
For instance, topological 
multi-sorting could be easily presented in undergraduate 
Discrete Mathematics or Data Structures and Algorithms
courses. We view this as a strength rather than
as a weakness. Only \emph{because} the results are simple and natural
do we entertain the hope to extend them to more complex cases.
In particular, we hope that the results developed here for \emph{string} rewriting can serve as a 
stepping stone for tackling the problem, left open in~\cite[Chapter~8]{Tere:03}
(see Remark~\ref{rem:klop}), of giving a characterisation of permutation equivalence 
for \emph{term} rewriting\footnote{%
But also other types of \emph{structured} rewrite systems 
such as graph rewrite systems come to mind.
} by \emph{tracing}.
Interesting (classes of) term rewrite systems to target are: 
\begin{itemize}
\item
  \emph{Linear} systems.
  In the first-order case linearity can be brought about 
  by requiring for each rule that every variable occurs 
  either zero or one time in \emph{both} its sides.
  The characterisation should cover not only string rewrite systems
  via their \emph{monadic} embedding (see Remark~\ref{rem:oudenadic}),
  and \emph{chemical} systems~\cite[Example~8.6.1]{Tere:03},
  but also rules like $\rulstr{x + 0}{x}$.\footnote{%
I expect the adaptation of the results to first-order linear 
term rewriting to be largely unproblematic but still interesting.
For instance, trace graphs for term rewrite systems 
cannot be planar, unlike those for string rewrite systems here.}
  In the higher-order case~\cite[Chapter~11]{Tere:03}
  linearity could be brought about by restricting to a \emph{linear}  
  substitution calculus~\cite{Smit:17}.
\item
  \emph{Non-linear} systems. As illustrated in the introduction, the problem
  in the non-linear case is that the replicating effect of a non-linear 
  term rewrite step is not represented in the term structure,
  so cannot be traced. One could hope this can be overcome by reifying replication
  (so it becomes \emph{traceable}), 
  by making the \emph{substitution calculus}~\cite{Oost:Raam:94} suitably explicit.
  \emph{Sharing graphs} as known from the theory of optimal 
  reduction~\cite{Aspe:Guer:98,Oost:Looi:Zwit:04} suggest themselves,
  in both the first- and higher-order cases, with 
  Combinatory Logic and the $\lambda\beta$-calculus, respectively,
  concrete systems to try potential characterisations on.
\end{itemize}

All our results are effective and constructive, but we did not study their complexity.
However, we do hope that the concrete representations of permutation equivalence classes
by means of tragrs (certain graphs) and greedy multistep reduction
(certain terms) could be useful for such, cf.\ \cite{Deho:15}.

\paragraph{
Acknowledgments}

%

We thank Jan Willem Klop for making the initial remark,
Nao Hirokawa and the reviewers and participants 
of Termgraph 2022 in Haifa for feedback, and the 
reviewers for thorough reading and many 
helpful comments and suggestions.


\bibliographystyle{eptcs}
\bibliography{eptcs-vvo-tragr}

\appendix
\section{Topologically multi-sorting the tragr of Figure~\ref{fig:tragr} (right)} \label{app:stages}

Reading back the tragr in Figure~\ref{fig:tragr} (right)
by topological multi-sorting gives rise to the following $6$ stages. 
\def\dfig{$\stremp$}%
\def\capfig{$1^\text{st}$ multistep: $\A\B\beta$}%
\def\ccpfig{$2^\text{nd}$ multistep: $\A\alpha\beta$}%
\def\cdpfig{$3^\text{rd}$ multistep: $\beta \A\A\B$}%
\def\cepfig{$4^\text{th}$ multistep: $\B\beta \A\A\B$}%
\def\cfpfig{$5^\text{th}$ multistep: $\alpha\beta \A\A\B$}%
\def\cgpfig{Last (empty) multistep: $\A\B\A\A\B\A\A\B$}%
\begin{center}
\scalebox{.29}{\input{tragrlayers1.txt}}\quad\quad%
\scalebox{.29}{\input{tragrlayers3.txt}}
\end{center}
\begin{center}
\scalebox{.29}{\input{tragrlayers4.txt}}\quad\quad%
\scalebox{.29}{\input{tragrlayers5.txt}}
\end{center}
\begin{center}
\scalebox{.29}{\input{tragrlayers6.txt}}\quad\quad%
\scalebox{.29}{\input{tragrlayers7.txt}}
\end{center}
Vertically composing the corresponding $5$ multisteps (not taking the 
last empty multistep into account) yields the multistep reduction
$\A\B\beta \cdot
 \A\alpha\beta \cdot
 \beta \A\A\B \cdot
 \B\beta \A\A\B \cdot
 \alpha\beta \A\A\B$. That is, we have read back $\acnv'$
 from the joint tragr in Figure~\ref{fig:tragr},
 of $\acnv$ and $\acnv'$ in Example~\ref{exa:srs}!
\end{document}